\documentclass{llncs}

\title{Automated Synthesis of a Finite Complexity Ordering for Saturation}
\author{Yannick Chevalier $^1$ ~~~~  Mounira Kourjieh $^2$}
\institute{$^1$ IRIT, Universit{\'e} de Toulouse, France\\ 
$^2$ LORIA-CNRS, Nancy, France}

\usepackage{color}
\usepackage{amsmath,amssymb}
\usepackage{proof}
\usepackage{xspace}
\usepackage{algorithmicx}
\usepackage{url}
\usepackage{subfig}

\newenvironment{proofclaim}{%
  \begin{list}{}{\setlength{\leftmargin}{1em}}%
\item\hspace*{-1em}\textbf{Proof of the claim.}\it}{%
  \end{list}}

\makeatletter

\def\make@math@constant#1{%
\expandafter\def\csname #1\endcsname{%
\ensuremath{\operatorname{#1}}}
}
\def\make@math@function#1{%
\expandafter\def\csname #1\endcsname##1{%
\ensuremath{\operatorname{#1}(##1)}}
}
\def\make@math@cal#1{%
\expandafter\def\csname #1\endcsname{%
\ensuremath{\mathcal{#1}}\xspace}
}

\def\make@math@ded#1{%
\expandafter\def\csname D#1\endcsname{%
\ensuremath{\mathcal{D}_{\textrm{\footnotesize #1}}}}
}

\def\make@math@iterator#1#2{\@for\p@arg:=#1\do{%
\edef\arg{\expandafter\@firstofone\p@arg}
\@expandtwoargs{\csname make@math@#2\endcsname}{\arg}\relax
}}

\def\MakeFunctions#1{\make@math@iterator{#1}{function}}
\def\MakeConstants#1{\make@math@iterator{#1}{constant}}
\def\MakeCal#1{\make@math@iterator{#1}{cal}}
\def\MakeDed#1{\make@math@iterator{#1}{ded}}
\makeatother

\MakeCal{X,C,F,P,T,R,I,D,S}

\MakeFunctions{ar,Unif,mgu,Var,atoms}
\newcommand{\TFX}{\ensuremath{\T(\F,\X)}\xspace}

\newcommand{\set}[1]{\ensuremath{\lbrace #1 \rbrace}}

\newcommand{\ded}{\leadsto}
\newcommand{\rew}{\to}

\newcommand{\atomordering}{{\succ_a}}
\newcommand{\termordering}{{\succ_t}}

\newcommand{\Predicate}{\ensuremath{\mathcal{P}}}

\newcommand{\arity}{\ensuremath{arity}}

\begin{document}
\sloppy
\maketitle

\begin{abstract}
  We present in this paper a new procedure to saturate a set of
  clauses with respect to  a well-founded ordering on ground atoms
  such that $A\prec B$ implies $Var{A}\subseteq\Var{B}$ 
  for every atoms $A$ and $B$. This
  condition is satisfied by any atom ordering compatible with a
  lexicographic, recursive, or multiset path ordering on terms. Our
  saturation procedure is based on a priori ordered resolution and its
  main novelty is the on-the-fly construction of a finite complexity
  atom ordering. 
  In contrast with the usual redundancy, we give a new redundancy notion and we
 prove that during the saturation a non-redundant inference by a
 priori ordered resolution is also an inference by a posteriori
 ordered resolution. We also prove that if a set $S$ of clauses  is
  saturated with respect to an atom ordering as described above then the problem
  of whether a clause $C$ is entailed from $S$ is decidable.
\end{abstract}

\section{Introduction}

Resolution is an inference rule introduced by Robinson~\cite{robinson}
for theorem proving in first-order logic. It consists in saturating a
theory presented by a finite set of disjunctions, called clauses, with
all its consequences.
%
Since the seminal work of Robinson,  lot of efforts have been devoted to
finding strategies that limit the possible inferences but still are
complete for refutation. The correctness of resolution implies the
correctness of these strategies. Among these we note selected
resolution~\cite{BachmairGLS95} and ordered resolution~\cite{Bachmair-Ganzinger90} which are
 correct and 
refutationally complete.
The later being
a special case of \cite{Rusinowitch86}. Later, it was  proved
in~\cite{Basin-Ganzinger-01} that if a set $S$ of clauses is saturated by
ordered resolution (with some additional hypotheses discussed in this
paper) then deciding whether a clause $C$ is a
consequence of $S$ is decidable. We present in this paper a weakening
of the hypotheses assumed in~\cite{Basin-Ganzinger-01} that also enjoys 
this decidability property.  In ~\cite{Basin-Ganzinger-01}, it is proved that saturated sets of clauses
are order local, and thus 
if each atom has a finite number of smaller atoms 
then the ground entailment problem is decidable.
Orders having this property are said to be of \emph{finite complexity}.

We present in this paper a variant of the standard saturation procedure
that builds during saturation an \emph{atom rewriting system}. This rewriting
system defines a partial ordering on ground atoms that has a finite complexity. 
Under our redundancy notion, we prove that the saturation of a set $S$ of clauses
implies its locality (as in \cite{Basin-Ganzinger-01}) with respect to 
the ordering based on the \textit{atom rewriting system}. As a consequence, if a set $S$ 
of clauses  is
  saturated with respect to an atom ordering as described above then the problem
  of whether a clause $C$ is entailed from $S$ is decidable. Finally we prove 
that the conditions imposed on the atom ordering  are satisfied by all atom 
ordering compatible with a  well-founded, monotone, and subterm term ordering,
\textit{i.e.}, most of the standard term orderings.

%

\paragraph{Outline of this paper.} 
In Section \ref{sec:preliminaries}, we present the basic notions that we will use later 
in the paper, in Section \ref{sec:old-decidability-result}
we present some of the decidability results for the \textit{ground entailment problem} given in the literature,
in Section \ref{sec:result} we present
our definitions  of \textit{atom rewriting system}, \textit{locality} and \textit{redundancy}, 
in Section 
\ref{subsec:sat:saturation} we give 
our \textit{saturation procedure},  in Section \ref{sec:decidability} we give 
our decidability result, and in  Section \ref{sec:comp}
we show 
 how our result extends the results presented in Section \ref{sec:old-decidability-result}.

\section{Formal setting}\label{sec:preliminaries}
\subsection{Basic notions}\label{subsec:basic-notions}
\paragraph{Syntax.}

We assume that we have an infinite set of variables $\X$, an infinite  set of constant symbols
 $\C$, a set of predicate symbols  $\Predicate$ and a set of function symbols $\F$. We associate  the function   
$\arity$ to function symbols and predicate symbols, 
$\arity: \F\cup\Predicate\rightarrow \mathbb{N}$.
The arity of a function symbol (respectively predicate symbol) indicates the number of arguments that 
the function symbol (respectively the  predicate symbol) expects.
We define the set of \textit{terms} $\TFX$ as follows:
$\X,\C\subseteq\TFX$, and for each function symbol $f\in\F$ with arity $n\geq 0$, for each terms 
$t_1,\ldots,t_n\in\TFX$, we have   
$f(t_1,\ldots,t_n)\in\TFX$.   We denote by $\Var{t}$ the set of variables 
occurring in the term  $t$, and a term $t$ is said to be \textit{ground} if $\Var{t}=\emptyset$.
We define subterms of a term $t$, denoted $Sub(t)$, as follows: if $t$ is a constant or a variable
then
$Sub(t)=\{t\}$, if $t=f(t_1,\ldots,t_n)$ then 
$Sub(t)=\{t\}\cup \bigcup_{i\in\{1,\ldots,n\}} Sub(t_i)$.
We denote by $t[s]$ a term $t$ containing $s$ as subterm.  
We define \textit{atoms}  as follows: if $I$ is a predicate symbol 
in $\Predicate$ with arity $n\geq 0$ and $t_1,\ldots,t_n$ are terms in $\TFX$ then $I(t_1,\ldots,t_n)$ is an atom.
A \textit{literal} $L$ is either  $A$ or $\neg A$ 
where $A$ is an atom and $\neg$ denotes the \textit{negation}. The literal $L$ is 
a positive literal in the first case, and a negative literal in the second. 
We denote by $\Var{A}$ the set of variables 
occurring in the atom $A$ and an atom  $A$ is said to be \textit{ground} if $\Var{A}=\emptyset$.

A \textit{clause} \textit{(or full clause)} is defined by  a set of literals 
$\{\neg A_1,\ldots,\neg A_m,B_1,\ldots, B_n\}$. 
It may also be viewed as a formula of the form
$\Gamma\rightarrow \Delta$ where $\Gamma=\{A_1,\ldots,A_m\}$ 
and $\Delta=\{B_1,\ldots,B_n\}$; 
$\Gamma$ represents the \textit{antecedent of the clause} and $\Delta$ its \textit{succedent}.
We denote $Atoms(C)$ the set of atoms occurring in the clause $C$.
A clause is \textit{ground} if all its atoms are ground. 
A clause $\Gamma\rightarrow \Delta$ is \textit{Horn} when $\Delta$ is a singleton or empty, and 
 \textit{unit} when it has only one literal.
A clause $\Gamma\rightarrow \Delta$ is \textit{positive} when it has only a succedent, 
$i.e.$ $\Gamma=\emptyset$ and  is \textit{negative} when it has only
an antecedent, $i.e.$ $\Delta=\emptyset$.  
We  write $\Gamma_1,\Gamma_2$ to indicate the union of the two sets 
$\Gamma_1$ and $\Gamma_2$
and usually omit braces. For example,
 we write $\Gamma,A$ or $A,\Gamma$ for the union of $\{A\}$ and $\Gamma$ or write
 $A_1,\ldots,A_m\rightarrow B_1,\ldots,B_n$ for 
$\{A_1,\ldots, A_m \}\rightarrow\{B_1,\ldots, B_n\}$.
We also make more simplifications, 
for example we write $A$ to denote  the positive unit clause 
$\emptyset\rightarrow A$, and $\neg A$ to denote the  negative unit clause $A\rightarrow\emptyset$.
Let $C$ be a clause,  we denote by $\neg C$ the set of unit clauses $\neg L$
with $L$ a literal in $C$; For example, $\neg C=\{A_1,\ldots, A_m,\neg B_1,\ldots, \neg B_n\}$ 
 when $C=A_1,\ldots,A_m\rightarrow B_1,\ldots, B_n$.
We say that a term $t$ occurs in an atom $A$ if $A$ is of the form $I(\ldots, u,\ldots)$ with $t$ a subterm of $u$ and 
$t$ occurs in a clause if it occurs in an atom of the clause.

\paragraph{Substitutions and unifiers.} 
A substitution $\sigma$ is a partial function   from variables $\X$ to terms 
$\TFX$ such that $Supp(\sigma)=\{ x|\sigma(x)\not=x\}$ is a finite set
 and  $Supp(\sigma)\cap Var(Ran(\sigma))=\emptyset$ with 
$Ran(\sigma)=\{\sigma(x) | x\in Supp(\sigma)\}$.
We denote by $Var(\sigma)$ the set $Var(Ran(\sigma))$.
A substitution $\sigma$ with $Supp(\sigma)=\emptyset$ is called the
\textit{empty substitution} or the \textit{identity substitution}.
A substitution $\sigma$ is said to be \textit{ground} 
if $Var(\sigma)=\emptyset$, that is $Ran(\sigma)$ is a set of ground terms. 
A \textit{renaming} $\rho$ is an injective substitution such that  $Ran(\rho) \subseteq \X$.
A  substitution $\sigma$ is \emph{more general} than a
substitution $\tau$, and we note $\sigma\le\tau$, if there exists a substitution $\theta$ such that
$\sigma\theta=\tau$. 
 Equivalent substitutions, \textit{i.e.}
substitutions $\sigma$ and $\tau$ such that $\sigma\le\tau$ and
$\tau\le\sigma$ are said to be equal up to renaming since in that
case the substitution $\theta$ is a bijective mapping from variables
to variables.
If $M$ is an  expression (i.e. a term, an atom, a clause, or a set of such objects)
and $\sigma$ is a substitution, then $M\sigma$  
is obtained by applying $\sigma$ to $M$ as usually defined; 
We say that $M\sigma$ is an \textit{instance} of $M$ and if $M\sigma$ is ground we say that $\sigma$ is 
\textit{grounding} $M$.

A substitution $\sigma$ is said to be a \emph{unifier} of two elements 
(\textit{i.e.} terms or atoms)
$e_1,e_2$ if  $e_1\sigma = e_2\sigma$. 
We denote $\Unif{e_1,e_2}$
the set of unifiers of the two elements $e_1$ and $e_2$.
  It is well-known that whenever the set $\Unif{e_1,e_2}$
is not empty it has a unique minimal element up to renaming. This minimal element is  called
the \emph{most general unifier} of $e_1$ and $e_2$, and is denoted
\mgu{e_1,e_2}.

\paragraph{Orderings.}
A (strict) ordering $\succ$ on a set of  elements  $E$ is a transitive and irreflexive binary relation on $E$.
The ordering $\succ$ is said to be:
\begin{itemize}
\item
\textit{total} if
for any two different elements $e,e'\in E$, we have either $e\succ e'$ or $e'\succ e$;
\item 
\textit{well-founded}  if   there is no infinite descending chain $e\succ e_1\succ  \ldots$ for any element  
$e$ in $E$; 
\item \textit{monotone} if $e\succ e'$ then $e\sigma \succ  e'\sigma$ for any  elements $e, e'$ in $E$ 
and any substitution $\sigma$.
\end{itemize}
Any ordering $\succ$ on a set of elements $E$ can be extended to an
ordering $\succ^{set}$ on finite sets over $E$ as follows: if $\eta_1$
and $\eta_2$ are two finite sets over $E$, we have
$\eta_1\succ^{set}\eta_2$ if (i) $\eta_1\not=\eta_2$ and (ii) for
every $e\in \eta_2\setminus \eta_1$ then there is
$e'\in\eta_1\setminus \eta_2$ such that $e'\succ e$.  
Given a set $\eta_1$, a
smaller set $\eta_2$ is obtained by replacing an element 
in $\eta_1$ by a (possibly empty)
finite set of strictly smaller elements.  We call an element $e$
\textit{maximal} (respectively \textit{strictly maximal}) with respect
to a set $\eta$ of elements if for any element $e'\in\eta$ we have
$e'\not\succ e$ (respectively $e'\not\succeq e$).  If the ordering
$\succ$ is total (respectively well-founded and  monotone),
so is its set extension.

We denote by an \textit{atom ordering} $\atomordering$ (respectively \textit{term ordering} $\termordering$)
any arbitrary  ordering on atoms (respectively on terms).  
We  extend an \textit{atom ordering} $\atomordering$ to a  \textit{clause ordering} as follows:
  we identify 
  clauses  with   the 
  sets of their respective atoms, and we order the clauses
 with respect to the sets of their respective atoms using the ordering $\atomordering^{set}$. 
For example, the clauses $A_1,A_2\rightarrow B$ and $A_1\rightarrow B$ are
 identified respectively with the  following
 sets of atoms $\{A_1,A_2,B\}$ and $\{A_1,B\}$;  
The second set is strictly smaller than the first one with respect to the 
 ordering   $\atomordering^{set}$, and hence the second clause is strictly smaller than the first one.

\begin{center}
  \fbox{\vspace*{8em} ~~\parbox{0.92\linewidth}{
In the remainder of this  paper, we   assume that
the atom ordering $\atomordering$ is  monotone, well-founded, and is such that 
$A\prec_{a} B$ implies $\Var{A}\subseteq \Var{B}$ for every atoms $A$ and $B$.
 }~~}
\end{center}

\subsection{Resolution}
The resolution is an inference rule   introduced by Robinson \cite{robinson};
It is one of the most successful methods for automated proof search in 
 first-order logic.
We say that a set $S$ of clauses is  \textit{unsatisfiable}
if there is no \textit{Herbrand interpretation}  satisfying it, 
and \textit{satisfiable} otherwise. 
Given a set $S$ of clauses and a ground clause $C$,
$S\models C$ means that 
$C$ is true in every Herbrand model of $S$;
It is easy to see that $S\models C$ iff $S\cup\neg C$ is unsatisfiable.
A \textit{proof by refutation} of $S\models C$ consists in proving that $S\cup\neg C$ is unsatisfiable.
The resolution has been proved in \cite{robinson} to be correct and complete for refutation.
The correctness of the resolution means that  the empty clause (\textit{i.e.} a contradiction) 
 can not be derived from any satisfiable set of clauses, and the completeness means that  
 the empty clause 
 can be derived from any unsatisfiable set of clauses. 

 The \textit{resolution}  is
 described by the  two inference rules given in  Fig.~\ref{fig:inferences}. 
 The clause $(\Gamma,\Gamma'\rightarrow
 \Delta,\Delta')\alpha$ of the resolution rule is called the
 \textit{resolvent} of the premises
  ($\Gamma\rightarrow \Delta,A$ and $A',\Gamma'\rightarrow \Delta'$) or the \textit{conclusion} of the
 inference, and the atom $A\alpha$ is called the \textit{resolved}
 atom. In the factoring rule, the clause $(\Gamma\rightarrow
 \Delta,A)\alpha$ is called the \textit{factor} of the premise 
 ($\Gamma\to \Delta,A,A'$)
or the
 \textit{conclusion} of the inference, and the atom $A\alpha$ is
 called the \textit{factored} atom.

\begin{figure}
  \centering
      \subfloat[\textit{Resolution rule.}]{\parbox{0.45\linewidth}{
        $$
        \infer{(\Gamma,\Gamma'\rightarrow \Delta,\Delta')\alpha}
        {\Gamma\rightarrow \Delta,A ~~~~~~~~ A',\Gamma'\rightarrow \Delta'}
        $$\footnotesize
        where $\alpha= \mgu{A,A'}$.}}\qquad\qquad
    \subfloat[\textit{Factoring  rule.}]{\parbox{0.44\linewidth}{
        $$\infer{(\Gamma\rightarrow \Delta,A)\alpha}{\Gamma\rightarrow \Delta,A, A'}
        $$\footnotesize
        where $\alpha= \mgu{A,A'}$.}}
    \caption{Standard resolution and factoring rules}
  \label{fig:inferences}
\end{figure}

\paragraph{Ordered resolution.} 
Since the seminal work of Robinson \cite{robinson}
lot of efforts have been devoted to
finding strategies that limit the possible inferences but still are
complete for refutation and correct; The correctness of these strategies 
is obtained from the correctness of the resolution.
 Among these strategies, there is 
the ordered resolution \cite{BachmairG90}
which is used in this paper and will be presented in this paragraph.


The ordered resolution makes use of an atom ordering $\atomordering$ and 
is described by 
two inference rules: \textit{ordered factoring rule} and \textit{ordered resolution rule}.
We  distinguish two types of 
 \textit{ordered resolution}: the  
\textit{posteriori ordered resolution} and the
\textit{priori ordered resolution}. 

\begin{figure}
  \centering
      \subfloat[Posteriori ordered resolution
    rule.]{\parbox{0.45\linewidth}{
        $$
        \infer{(\Gamma,\Gamma'\rightarrow \Delta,\Delta')\alpha}{\Gamma\rightarrow \Delta,A 
          ~~~~~~~~ A',\Gamma'\rightarrow \Delta'}
        $$\footnotesize
        where $\alpha= \mgu{A,A'}$,
        $A\alpha$ is strictly maximal with respect to $\Gamma\alpha, ~
        \Delta\alpha$ for $\atomordering$, and $A\alpha$ is maximal
        with respect to $\Gamma'\alpha, ~ \Delta'\alpha$ for
        $\atomordering$.}}\qquad\qquad
    \subfloat[\textit{Posteriori ordered factoring
      rule.}]{\parbox{0.44\linewidth}{
        $$\infer{(\Gamma\rightarrow \Delta,A)\alpha}{\Gamma\rightarrow \Delta,A,A'}
        $$\footnotesize
        where $\alpha= \mgu{A,A'}$, $A\alpha$ is strictly maximal with
        respect to $\Gamma\alpha$ for $\atomordering$, and maximal
        with respect to $\Delta\alpha$ for $\atomordering$.}}

    \subfloat[\textit{Priori ordered resolution  rule.}]{\parbox{0.45\linewidth}{$$
        \infer{(\Gamma,\Gamma'\rightarrow
          \Delta,\Delta')\alpha}{\Gamma\rightarrow \Delta,A ~~~~~~~~
          A',\Gamma'\rightarrow \Delta'}
        $$\footnotesize
        where $\alpha= \mgu{A,A'}$, $A$ is maximal with respect to
        $\Gamma, ~ \Delta$ for $\atomordering$, and $A'$ is maximal with
        respect to $\Gamma', ~ \Delta'$ for $\atomordering$.}}\qquad\qquad
    \subfloat[\textit{Priori ordered factoring rule.}]{\parbox{0.44\linewidth}{$$
        \infer{(\Gamma\rightarrow \Delta,A)\alpha}{\Gamma\rightarrow \Delta,A,A'}
        $$\footnotesize 
        where $\alpha= \mgu{A,A'}$, $A$ is maximal with respect to $\Gamma$
        and $\Delta$ for $\atomordering$.}}
    \caption{Posteriori and priori ordered resolution and factoring
      rules.}
  \label{fig:ordered:inferences}
\end{figure}

\paragraph{Remarks.} 
\begin{enumerate}
\item
 The \textit{posteriori} ordered resolution is 
actually the \textit{ordered} resolution introduced in \cite{BachmairG90} and 
the \textit{priori} ordered resolution is related to the  \textit{selected resolution}
which is widely studied in the literature \cite{Lynch03}.
\item
We remark  that the two types of ordered resolution  coincide on ground
clauses, but  not on  non-ground clauses:
 let us consider the following two clauses
$C=I(b,y)\rightarrow I(x,y)$ and $D= I(a,b)\rightarrow \emptyset$ and the ordering:
$I(a,b)\prec_a I(b,b)$.
We have $\sigma=\{x \mapsto a, y\mapsto b\}$
is the most general unifier of $I(x,y)$ and $I(a,b)$. We remark that 
$I(a,b)$
is maximal with respect to $\emptyset$, 
$I(x,y)$ and  $I(b,y)$ are not comparable and hence $I(x,y)$
is  maximal with respect to  $I(b,y)$. This implies that the \textit{priori} ordered
resolution inference rule can be applied to the clauses $C$ and $D$ but not the 
\textit{posteriori} ordered inference rule  since 
$I(x,y)\sigma=I(a,b) \prec_a I(b,b)$.
We  remark that  
in the case of monotone atom ordering as we consider in this paper, 
the \textit{posteriori ordered resolution} is included in the \textit{priori ordered resolution}.
\end{enumerate}
In spite of this difference between  priori and  posteriori 
ordered resolution, we introduce a  redundancy notion such that every non-redundant priori ordered resolution inference 
is a posteriori ordered resolution inference (see Lemma \ref{lem:redundancy}).

It is well-known that the \textit{posteriori} ordered resolution and the \textit{priori}
ordered resolution are correct and complete \cite{robinson,BachmairG90}. 

\paragraph{Ground entailment problem.}
The \textit{ground entailment problem} 
studied in this paper is defined as follow:
\begin{center}
   \fbox{\vspace*{8em} ~~\parbox{0.92\linewidth}{
 Given a set $S$ of clauses, the \textit{ground entailment problem for $S$} is defined as follows: 
\begin{description}
 \item \textbf{Input:}  a ground clause  $C$.
 \item \textbf{Output:}  "entailed" if and only if  $S\models C$. 
 \end{description}
 }~~}
\end{center}

\section{Decidable fragments of first order logic}\label{sec:old-decidability-result}
It is known that the ground entailment problem for Horn clauses and full clauses sets is undecidable
in general. 
Here, we mention  decidability results for some fragments. 
  
\subsection{McAllester's result}
In \cite{McAllester-93},
D. McAllester was interested by  Horn clauses.
He first defined the  subterm locality as follows: 
 a set $S$ of Horn clauses  is \textit{subterm local} 
if for every ground Horn clause $C$, we have $S\models C$ if and only if 
 $C$ is entailed from a set of ground instances of clauses in $S$
in which each term is a subterm of a ground term in $S$ or in  $C$.
It is  proved in \cite{McAllester-93} that if a set $S$ of Horn clauses
is finite and subterm local then its ground entailment problem is  decidable.
 
\subsection{Basin and Ganzinger results}
In \cite{Basin-Ganzinger-01}, D. Basin and H. Ganzinger  generalized McAllester's result
 by
 allowing monotone, total,  well-founded ordering 
over terms, 
and full (not Horn) clauses. 
To this end, they introduced several notions and we recall next some of them.
A set of clauses $S$ is said to be \textit{order local} with respect to a term ordering $\succ_t$ 
if for every ground clause $C$, we have $S\models C$
if and only if $C$ is entailed from  a set of ground instances of clauses in $S$
in which each  term is smaller than or equal to some  term in $C$.
It is proved in \cite{Basin-Ganzinger-01} 
that if a set $S$ of clauses 
is saturated up to redundancy by posteriori ordered resolution for 
a total, monotone, well-founded atom ordering then $S$ is order local.

A term ordering $\succ_t$ is said to be of complexity $f,g,$ whenever for each clause of size 
$n$ (the size of a  term is the   number of nodes in its tree representation and the size of
 a  clause is the sum of sizes of its terms) there exists $O(f(n))$ terms that are smaller than or
 equal  to a term in the clause, and that they  may be enumerated in time $g(n)$.
D. Basin and H. Ganzinger obtained the following decidability results:
\begin{enumerate}
\item
If $S$ is a set of (full) clauses that is order local with respect to a term ordering $\succ_t$ of 
complexity $f,g$ then the ground entailment problem for $S$ is decidable. 
\item
 If $S$ is a  set of (full) clauses saturated up to redundancy  by posteriori ordered resolution with respect 
to a complete  well-founded atom ordering, and if, for each clause in $S$, each of its maximal atoms 
contains all the variables of the clause, then the ground entailment problem for $S$ is decidable.
\item
However, if the restriction on the variables
in maximal atoms (the condition in the previous  point) is removed, the ground entailment problem becomes undecidable in general.
\end{enumerate}
We prove in this paper that it is possible to partially remove the condition on variables mentioned above
while keeping the decidability of the ground entailment problem.
More precisely: 
we prove in Theorem 1  the decidability of the ground entailment problem
for  $S$ 
when 
$S$ is a finite saturated set of clauses with respect to  an atom ordering which is  
  well-founded, monotone and  such
  that $A\prec_a B$
  implies $\Var{A}\subseteq\Var{B}$ for every atoms $A$ and $B$.

The next three  sections are  devoted to  this result.

\section{Locality and redundancy}\label{sec:result}
We introduce an atom rewriting system  to model a new ordering relation
between atoms. 
Our goal is to restrict the atom ordering $\prec_a$
to an ordering $\prec_\R$ such that 
each atom has only a finite number of predecessors.

\begin{definition}(Rewriting system on atoms.) 
 Given an atom ordering $\atomordering$, we define a 
  \emph{rewriting system \R on atoms} as  a
  set of rules $L\rew R$ where $L$ and $R$ are two atoms with
  $L\succeq_a R$.
\end{definition}

We give next some definitions that we  use later in this section.
\begin{definition}\label{defrelations}
Let $A$ and $B$ be two  atoms,  $C$ be a clause and $\R$ a rewriting system on atoms.  We have:
\begin{itemize}
\item 
 $A\downarrow_{\R}=\{B \text { such that } A\rew^*_{\R} B \}$, \textit{i.e.}
$A\downarrow_{\R}$ denotes the set of atoms reachable from $A$ when applying
rules in \R.
\item
$C\downarrow_{\R}=\{A\downarrow_{\R} \text { such that  $A$ is an atom in $C$ }\}$.
\item
$A\downarrow_{{\R}^-}=A\downarrow_{\R}\setminus\set{A}$. 
\item
$C\downarrow_{{\R}^-}=\{A\downarrow_{{\R}^-} \text { such that   $A$ is an atom  in $C$ }\}$.
\item
$A\prec_{\R} B$ whenever $A\in B\downarrow_{{\R}^-}$.
\end{itemize}\end{definition}

\begin{lemma}\label{PropertiesOfR}
Let $A$ and $B$ be two distinct  atoms.
 We have that   $A\rew_{\R} B$  implies $A \atomordering B$; And  
$A\prec_{\R} B$  implies  $A\prec_a B$ and $\Var A\subseteq\Var B$.\end{lemma}
\begin{proof}
Let $A$ and $B$
be two  distinct atoms such that $A\rew_{\R} B$, 
then there exists a rule $L\to R\in \R$, a substitution $\sigma$
such that $A=L\sigma$ and $B=R\sigma$. 
By definition of $\R$, we have $L \succeq_a R$ and then, by monotonicity of $\succ_a$, 
$L\sigma=A \succeq_a R\sigma=B$.
Since $A$ and $B$ are different, we conclude that $A\succ_a B$. 
Now we assume that $A\prec_{\R} B$, this implies that 
 $A\in B\downarrow_{{\R}^-}$,
and hence $B\rew^*_{\R} A$. Since $A\not= B$ we then have  $B\atomordering A$.
 Since  $A\prec_a B$ implies $\Var A \subseteq \Var B$ (by hypothesis on the ordering $\succ_a$), we then have 
  $A\prec_{\R} B$ implies $\Var A \subseteq \Var B$.
\end{proof}

\begin{lemma}\label{lem:reach-is-finite} 
  Let $\R$ be  a finite  rewriting system on atoms.
 If $A$ is a  ground atom then the set $ A\downarrow_{\R}$ is finite.
\end{lemma}
\begin{proof} 
Let $A$ be a ground atom. 
By Lemma \ref{PropertiesOfR}, we have $A\downarrow_{\R}$ is a set of ground atoms.
Consider that graph $G=(A\downarrow_{\R},V)$ 
where $(D,D')\in V$ if and only if $D\not=D'$
and $D\to_{\R} D'$.
By Definition \ref{defrelations},
 $(D,D')\in V$ implies $D\succ_{\R} D'$. Thus $G$ is acyclic.
Since $\R$ is finite and 
$\Var{R}\subseteq\Var{L}$ for every rule $L\to R\in\R$,
 each node has a finite number of direct successor 
nodes. 
By \textit{$K\ddot{o}nig's$} lemma, 
if the graph $G$ is infinite it has an infinite path. The atoms on this infinite path 
form an infinite strictly decreasing 
sequence of atoms $A\succ_a A_1\succ_a A_2 \succ_a \ldots$
which contradicts the well-foundness of $\succ_a$. 
We then conclude that the graph $G$ is finite, and hence is the set $A\downarrow_\R$. 
\end{proof}

\begin{definition}{(Rewriting system based on a set of
    clauses)\label{def:rewriting:clause}}
  Let $S$ be a  set of clauses. The  rewriting system 
  $\R(S)$ based on 
$S$ is a rewriting system on atoms defined by  the set of rewriting rules $L\rew R$ such that
$L$ and  $R$ are two atoms of $C$ with $C\in S$ and  $L \succeq_a R$.
 \end{definition}
We remark that when  $S$ is finite $\R(S)$ is also finite,
and $S\subseteq S'$ implies 
$\R(S)\subseteq\R(S')$.

We now deviate from the traditional notion of \textit{refutational proof} and define instead the notion of 
\textit{local dag proof}. 
Informally, a \textit{refutational proof} of $S\cup \neg C$ 
for a set $S$ of clauses and a clause $C$ is a tree where leaves are labeled by ground instances of 
clauses in $\{S\cup\neg C\}$, internal nodes 
are labeled by the conclusion of the  resolution applied to the antecedent nodes, and the root is 
labeled by the empty clause. 
In the \textit{dag proof} we introduce an ordering on the nodes such that 
the leaves are minimal and the root is maximal with respect to this new ordering.  

\begin{definition}{(Dag proofs)\label{def:dag:proof}}
Given a set $S$ of clauses, a  clause $C$ and an ordered finite set  of ground clauses $(T,<_T)$.
We  call $(T,<_T)$  a \textit{dag proof} of $S\cup \neg C$ 
if:
\begin{enumerate}
\item
 for any clause $t\in T$,  we have 
  either   $t$ is a ground instance of a clause in $S\cup \neg C$, or 
  there exists $t_1,t_2\in T$ with  $t_1,t_2<_T t$
     and $t$ is the conclusion of the resolution applied to $t_1$ and $t_2$.
\item $T$ 
   contains the empty clause.
\end{enumerate}
When such $(T,<_T)$ exists,  we write  $S \vdash C$.
In a dag proof, each minimal clause with respect to the ordering $<_T$ is called a \textit{leave}. 
\end{definition}

\begin{definition}{(Local dag proofs)\label{def:local}}
Given a set $S$  of clauses,  a clause $C$, an ordered finite set  of ground clauses $(T,<_T)$ 
and a set $\mathcal{A}$  of ground atoms. 
We  call $(T,<_T)$  a \textit{$\mathcal{A}$-local dag proof} of $S\cup \neg C$ 
if $(T,<_T)$ is a \textit{dag proof} of $S\cup \neg C$ and   $Atoms(T)\subseteq \mathcal{A}$.
When such $(T,<_T)$ and $\mathcal{A}$ exist,  we write  $S \vdash_{\mathcal{A}} C$.
\end{definition}

\begin{lemma}{\label{lem:decidable:local:entailment}}
 Given a finite set $S$ of clauses, a ground clause $C$ and a finite  rewriting system on atoms  
  $\R$, we can decide whether  $S\vdash_{C\downarrow_\R} C$.
\end{lemma}
\begin{proof}
$\R$ is finite, and $C$ is ground, this implies that  $C\downarrow_\R$ is finite and ground
(Lemma \ref{lem:reach-is-finite}).
For each $C\downarrow_{\R}$ local dag proof of $S\cup\neg C$, leaves are in a finite set of ground clauses,
and the set of these leaves is unsatisfiable.
The problem consisting is determining 
whether a finite set of ground clauses is unsatisfiable is decidable, and hence 
we can decide whether there exists a $C\downarrow_\R$ 
local dag proof of $S\cup\neg C$. 
\end{proof}


We define a notion of \textit{redundancy}  that identifies clauses and inferences that are not needed for performing 
the saturation procedure.

\begin{definition}{(Redundancy)\label{def:redundancy}}
  Let \R be a finite rewriting system on atoms,
a  ground clause $C$ is called \textit{\R-redundant} in a set  $S$ of
    clauses  if  $S\vdash_{C\downarrow_\R} C$, 
 a non-ground clause $C$ is called  \textit{\R-redundant} in a set $S$ of
    clauses  if all its ground  instances are \R-redundant in $S$,
and an  inference $C',C"\ded C$ by ordered resolution 
    is called \textit{\R-redundant} in the set $S$  of clauses  if $(1)$ one of the premises ($C'$ and $C"$) is 
    \R-redundant in $S$, or else if $(2)$ $S\vdash_{C\downarrow_{\R}} C$.
\end{definition}
Note that under this definition of redundancy,
if a clause $C$ in $S$ is subsumed by a clause 
$C'$ in $S$ then $C$ is \R-redundant in $S$.
\\
%
Using this  notion of redundancy, we show next how to  relate \textit{a priori}
and \textit{a posteriori} ordered resolution  rules.

\begin{lemma}{\label{lem:redundancy}}
  Let $C_1=\Gamma_1 \rew \Delta_1, A_1$ and  
$C_2=A_2,\Gamma_2\rew \Delta_2$ be two clauses,
 $C_1, C_2 \ded C$  be  an inference by
 \textit{priori} ordered resolution with $A_1\sigma$ the resolved atom, and 
 $\R= \R(C_1\sigma)\cup
  \R(C_2\sigma)$.  Then  either this inference is \R-redundant in $\{C_1,C_2\}$ or is an
  inference by \textit{posteriori} ordered resolution.
\end{lemma}
\begin{proof} We have 
 $C_1=\Gamma_1 \rew \Delta_1, A_1$, 
$C_2=A_2,\Gamma_2\rew \Delta_2$, and $C_1, C_2 \ded C$ with 
$C=\Gamma_1\sigma,\Gamma_2\sigma \rew \Delta_1\sigma,\Delta_2\sigma$  be  an inference by
\textit{priori} ordered resolution.  
We assume that $C_1,C_2\ded C$  is not an inference
 by \textit{posteriori} ordered resolution. Then either 
  $A_1\sigma$ is not strictly maximal for
  $\atomordering$ in the set of atoms of $C_1\sigma$,
or $A_1\sigma$ is not maximal for $\atomordering$ in the set of atoms of $C_2\sigma$.
This implies that there is an atom $B$
in $C$ with 
$A_1\sigma \preceq_a B$.
Let $j$ be such that $B\in Atoms(C_j\sigma)$.
$C_j\sigma$ contains $A_1\sigma$ and $B$ with $A_1\sigma\preceq_a B$.
This implies that $B\to A_1\sigma\in\R(C_j\sigma)$,
and hence $A_1\sigma\sigma_{S,C}\in Atoms(C\sigma_{S,C})\downarrow_{\R}$
with $S=\{C_1,C_2\}$.
We then have  ${C_1,C_2}\vdash_{C\sigma_{S,C}\downarrow_{\R}} C$, and hence the inference  $C_1, C_2 \ded C$ 
is $\R$-redundant in $\{C_1,C_2\}$.
\end{proof}

\section{Saturation}
\label{subsec:sat:saturation}
\begin{definition}{(Saturated set of clauses)\label{def:saturated}}
  Let \R be a  rewriting system on atoms. We say that a set $S$ of clauses 
  is \textit{\R-saturated} up to redundancy by ordered resolution
  if $(1)$ any inference by priori ordered resolution from
  premises in $S$ is \R-redundant in $S$, $(2)$
  $\R(S)\subseteq\R$, and $(3)$
  for each \textit{priori} ordered resolution inference
  $C_1, C_2 \ded C$ with  $C_1,C_2\in S$,
  if the resolved atom $A\sigma$ is not
  strictly maximal in $C_1\sigma$ or 
not maximal in $C_2\sigma$
then 
  $\R(\{C_1\sigma,C_2\sigma\})\subseteq\R$.
\end{definition}

We  present now a procedure that,  providing it terminates, 
constructs from  a finite set $S$ of
clauses a pair  $(S',\R)$ such that 
 $S'$ is a 
 finite set
of clauses,   $\R$ is a  rewriting system on atoms, and 
 for every ground clause $C$, 
we have 
$S \models C$ iff
 $S'\vdash_{C\downarrow_\R}C$.

\begin{figure}
  \centering
\parbox{0.92\linewidth}{ 
        \begin{description}
        \item[Input:]~\\
          A finite set $S$ of clauses.
        \item[Initialization:]~\\
          Let $(S_1,\R_1)=(S,\R(S))$, and $i=1$.
\item[Transformation step:]~\\
We construct the pair  $(S_{i+1},\R_{i+1})$ from
the pair  $(S_{i},\R_{i})$ as follows:
Let $C_1,C_2\ded C$ be 
an inference by  ordered resolution with 
$C_1,C_2\in S_i$, and 
  $A\sigma$ the resolved atom; 
One of the following three  cases will be applied: 
  \begin{itemize}
  \item \textit{Non-maximality:}  If $A\sigma$ is not  strictly maximal for
    $\atomordering$ in the atoms of $C_1\sigma$ or not maximal for $\atomordering$ 
in the atoms of $C_2\sigma$
    then $S_{i+1}=S_i$, $\R_{i+1}=\R_i\cup\R(\set{C_1\sigma,C_2\sigma})$, and $i=i+1$;
  \item \textit{Redundancy:}  Otherwise, if
    $S_{i}\vdash_{C\downarrow_{\R_i}} C$
    then $S_{i+1}=S_i$, $\R_{i+1}=\R_i$, and $i=i+1$;
  \item \textit{Discovery:} Otherwise a new clause useful for establishing
    local proofs has been discovered, and hence 
    $S_{i+1}=S_{i}\cup\set{C}$,
  $\R_{i+1}=\R_i\cup\R(C)$, and $i=i+1$.
\end{itemize}
\item[Iteration:]~\\
  We repeat the \textbf{Transformation step} until a fixed point is
  reached.
\end{description}
Returns $(S_i,\R_i)$.  }
\caption{Saturation procedure }
\end{figure}
\begin{definition}
The saturation procedure is  called \emph{fair}  when every
possible inference by \textit{priori} ordered resolution has been  performed. \end{definition}
From now on, we only consider fair saturation procedure and we may omit the 
 word "fair" for simplicity.
\\
We prove next that the  saturation procedure  actually constructs a saturated
set of clauses.

\begin{proposition}{\label{prop:saturation:correct}}
  Let $S$ be a finite set of clauses and $(S',\R)$ 
  be the output of the saturation procedure. 
 $S'$ is \R{}-saturated.
\end{proposition}
\begin{proof}
Let $S$ be a finite set of clauses such that the saturation 
procedure terminates and outputs $(S',\R)$.
By the initialization and discovery cases of the saturation, we have  
$\R(S')\subseteq \R$,  and by the non-maximality case of the saturation procedure 
 we have  $\R(\set{C_1\sigma,C_2\sigma})\subseteq \R$ for each $C_1,C_2\in S'$ 
 on which \textit{priori} ordered resolution is possible
but not  \textit{posteriori} ordered resolution.
Now, we prove that any inference by ordered resolution from premises in $S'$ is \R-redundant in $S'$.
Let $C_1,C_2\ded C$ be an inference by ordered resolution with $C_1,C_2\in S'$.
 Since the saturation is fair, this inference has been considered during the computation of $(S',\R)$, and falls into one 
 of the following cases: 
the \textit{redundancy}, the {non-maximality},
the \textit{discovery}. 
By contradiction, assume  that    
 $C_1,C_2\ded C$ is not
  \R-redundant in $S'$, then  we fall in one of the two other cases: 
  \begin{description}
  \item[non-maximality:]   the resolved atom $A\sigma$ is not strictly maximal
    in the atoms of $C$. Therefore  
    $C_1,C_2\ded C$ is not
    an inference by \textit{posteriori} ordered resolution, 
 and  by construction 
$\R(C_1\sigma) \cup \R (C_2\sigma)     \subseteq \R$.
Furthermore, Lemma \ref{lem:redundancy} implies that  the inference is 
$\R(C_1\sigma)\cup\R(C_2\sigma)$-redundant,
and hence  it is \R-redundant, which contradicts our assumption of non-redundancy.   
  \item[discovery:] this case implies that $C\in S'$, and  then it is trivial that 
   the inference is
    \R-redundant in $S'$, which contradicts our assumption of non-redundancy.   
  \end{description}
  As a consequence every inference between two clauses of $S'$
  must be $\R$-redundant. We finally conclude that  $S'$ is \R-saturated.
\end{proof}

\section{Decidability of the ground entailment problem}
\label{sec:decidability}
We consider in this section a finite set $S$ of clauses, and a finite rewriting system \R on atoms  
such that $S$ is \R{}-saturated. 
\begin{proposition}{\label{prop:locality}}
  Let $C$ be a ground
  clause. We have that $S\models C$ implies $S\vdash_{C\downarrow_\R} C$.
\end{proposition}
\begin{proof}
Let $\R$ be a finite rewrite system on atoms,   $S$  
be a finite set of clauses which is $\R$-saturated, and $C$ be a ground clause such  that
$S\models C$. Let $\Pi$ be a set of DAG proofs 
of $S\cup \neg C$. Since the resolution is complete and correct, we have  $\Pi\not=\emptyset$.
For every $\pi\in\Pi$, let   $\delta(\pi)=Atoms(\pi)\downarrow_{\R}\setminus Atoms(C)\downarrow_{\R}$
be the distance from $\pi$ to a local dag proof (if $\delta(\pi)=\emptyset$ then $\pi$
is a local dag proof).

Let $\pi\in\Pi$ be such that $\delta(\pi)$ is minimal, and let us prove that  $\delta(\pi)=\emptyset$.
By contradiction,  assume that $\delta(\pi)\not=\emptyset$
and let $A$ be a maximal atom in $\delta(\pi)$ for the ordering  $\prec_a$.
By Lemma \ref{PropertiesOfR}, we have that $B\to_{\R} B'$ implies that $B\succeq_a B'$ and hence 
$A$ is an atom of $\pi$.
We prove in the next \textit{claim} that $A$ must be  maximal with respect to the 
atoms of $\pi$ for the ordering $\prec_{\R}$.

\paragraph{\textbf{Claim 1.}}
   The atom $A$ is maximal in $Atoms(\pi)$ for
    the ordering $\prec_\R$.
  \begin{proofclaim}
    By contradiction if this were not the case there would exist an atom $B\in Atoms(\pi)$
    with $B\not= A$, $A\prec_\R B$, and hence $A\prec_a B$ (Lemma \ref{PropertiesOfR}). Since $A$ is maximal in
    $\delta(\pi)$ for the ordering $\prec_a$, we would have that $B$ is not in 
$Atoms(\pi)\downarrow_{\R}\setminus Atoms(C)\downarrow_{\R}$, and thus 
      $B\in
    Atoms(C)\downarrow_\R$. Since $A\prec_a B$, we have  that $A\in
    Atoms(C)\downarrow_\R$, which  contradicts $A\in \delta(\pi)$.
  \end{proofclaim}

Let $\text{\it Leaves}_A^+$ be the set of leaves of $\pi$ that contain the atom $A$, and
$\text{\it Leaves}_A^-$ be the set of leaves   that do not contain $A$.  
  The correctness and completeness of the  resolution
  implies that the set of clauses $\text{\it Leaves}_A^+\cup \text{\it
    Leaves}_A^-$ is unsatisfiable.

  \paragraph{\textbf{Claim 2.}}
    Each clause $C_A\in \text{\it Leaves}_A^+$ is an instance with a
    substitution $\sigma$ of a clause $C^s_A\in S$ with every atom $A^s$
    satisfying $A^s\sigma=A$ is maximal for $\atomordering$.

  \begin{proofclaim}
    By definition of $\text{\it Leaves}_A^+$,   $C_A$ is either a ground  instance of a clause in $S$
    or a clause in $\neg C$. Since $A$ is not an atom occurring in
    $C$ the later case is excluded. Thus there exists a clause $C^s_A\in S$,
    an atom $A^s\in C^s_A$, and a substitution $\sigma$ such that
    $A^s\sigma=A$ and $C^s_A\sigma=C_A$. Finally if $A^s$ is not
    maximal for $\atomordering$ in $C^s_A$ and $\R(S)\subseteq \R$ then it is not maximal for
    $\prec_\R$ in $C^s_A$ and thus by monotonicity, $A$ is not  maximal for $\prec_\R$ in the
    atoms of $C_A$. This contradicts the fact that $A$ is maximal
    for $\prec_\R$ among the atoms occurring in $\pi$.
  \end{proofclaim}
Thus every resolution on $A$ between two 
clauses $C,C'$ in $\text{\it Leaves}_A^+$ is a ground  instance with 
substitution $\sigma$ of a \textit{priori} ordered resolution between two clauses 
$C^s,C'^s$ in $S$.
In $\pi$, $\text{\it Leaves}_A^+$ are the unique leaves containing $A$; Furthermore, 
$A$ is maximal in each clause of $\text{\it Leaves}_A^+$ for the ordering 
$\atomordering$ by \textit{Claim 2}. 
This implies that we can first eliminate all the occurrences of the atom $A$ 
by application of the priori ordered resolution on 
$\text{\it Leaves}_A^+$, and let $\text{\it Leaves}'$ be the obtained set of clauses after 
performing all possible resolutions on $A$ in $\text{\it Leaves}_A^+$. 
The unsatisfiability of  $\text{\it Leaves}_A^+\cup \text{\it Leaves}_A^-$ implies
that unsatisfiability of $\text{\it Leaves}'\cup \text{\it Leaves}_A^-$.
We prove next that we can construct a new DAG proof $\pi'$ of $S\cup \neg C$ 
with $\delta(\pi')\prec_a^{set}\delta(\pi)$.

Let $C=\Gamma\to \Delta,A$, $C'=A,\Gamma'\to \Delta'$ be two clauses in $\text{\it Leaves}_A^+$, 
and let $C''$ be the result of the resolution on $C$ and $C'$. By  definition 
 of $\text{\it Leaves}'$, we have $C''\in\text{\it Leaves}'$. Let 
$C^s$ and $C'^s$ be two clauses on $S$ such that:
$C=C^s\sigma$, $C'=C'^s\sigma$, 
$A^s\in Atoms(C^s)$ with $A^s$ maximal 
in $C^s$ for $\atomordering$ and $A=A^s\sigma$, 
$A'^s\in Atoms(C'^s)$ with $A'^s$ maximal 
in $C'^s$ for $\atomordering$ and $A=A'^s\sigma$.
Wlog, assume that ordered factorization has been applied to
$C^s$ and $C'^s$  so that there is one-to-one mapping
between atoms of $C^s$ (respectively $C'^s$) and 
atoms of $C$ (respectively atoms of $C'$).
The \textit{priori} ordered resolution can then be applied to $C^s$ and $C'^s$ 
with $\theta$ the most general unifier of $A^s$ and $A'^s$, and $C''^s$ is the obtained clause.
Since $S$ is $\R$-saturated, this inference is saturated and then one of the two following cases holds:
\begin{enumerate}
\item $C''^s\in S$: $C''$
is then a ground instance of a clause in $S$. In this case we let $S(C'')=\{C''\}$. 
We remark that $Atoms(C'') =Atoms(C \cup C')\setminus \{A\}\subseteq Atoms(\pi) \setminus \{A\}$.
\item $C''^s\notin S$: by the \textit{saturation procedure}, we have two cases:
\begin{enumerate}
\item Non-maximality:  the inference $C^s, C'^s \ded C''^s$ is not an inference by \textit{posteriori}
 ordered resolution, and hence by Lemma \ref{lem:redundancy},
 the inference is $\R(C^s\theta\cup C'^s\theta)$-redundant in $\{C^s,C'^s\}$,
  and then $\R$-redundant in $\{C^s,C'^s\}$.
By definition of the \textit{redundancy}, we then have 
$S\vdash_{C''\downarrow_{\R}} C''$. We then let 
 $S(C'')$ be a set ground instances of clauses of $S$ whose
    atoms are in $C''\downarrow_\R$ that entails
    $C''$. 
We remark that $Atoms(S(C''))\subseteq C''\downarrow_\R
\subseteq (Atoms(\pi)\setminus A)\downarrow_\R$.
\item Redundancy: $C''^s$ is $\R$-redundant in $S$, and then by Definition  \ref{def:redundancy},
all ground instances of $C''^s$ are $\R$-redundant in $S$. This implies that 
 $S\vdash_{C''\downarrow_{\R}} C''$. We let 
  $S(C'')$ be a set ground instances of clauses of $S$ whose
    atoms are in $C''\downarrow_\R$ that entails
    $C''$.
We remark that $Atoms(S(C''))\subseteq C''\downarrow_\R
\subseteq Atoms(\pi)\setminus A\downarrow_\R$.
\end{enumerate}
\end{enumerate} 
The unsatisfiability of $\text{\it Leaves}_A^+ \cup \text{\it Leaves}_A^-$
implies the unsatisfiability of 
$\text{\it Leaves}_A^- \cup \bigcup_{C''\in\text{\it Leaves}'} S(C'')$,
 and hence there is a DAG proof $\pi'$ of 
$\text{\it Leaves}_A^- \cup \bigcup_{C''\in\text{\it Leaves}'} S(C'')$, which is also a DAG proof 
of $S\cup \neg C$.
 We prove next that $\delta(\pi')\prec_a^{set}\delta(\pi)$.

\begin{equation*}
 \begin{split}
\delta(\pi') &= Atoms(\pi')\downarrow_{\R}\setminus Atoms(C)\downarrow_{\R}\\
&= [Atoms(\text{\it Leaves}_A^-) \cup \bigcup_{C''\in\text{\it Leaves}'} Atoms(S(C''))]\downarrow_{\R}\setminus Atoms(C)\downarrow_{\R}\\
\\
&\subseteq (Atoms(\pi)\setminus A \cup Atoms(\text{\it Leaves}'))\downarrow_{\R}
\setminus Atoms(C)\downarrow_{\R}
\\
&\subseteq (Atoms(\pi)\downarrow_{\R}\setminus A) \setminus Atoms(C)\downarrow_{\R} (\text{maximality of }A\text{ in }Atoms(\pi))
\\
&= (Atoms(\pi)\downarrow_{\R}\setminus Atoms(C)\downarrow_{\R})\setminus A
\\
&=\delta(\pi)\setminus A.
 \end{split}
\end{equation*}
%
%
%
Since $A\in\delta(\pi)$ and is maximal, we then have 
 $\delta(\pi')\prec_a^{set}\delta(\pi)$, and hence, there is a DAG proof $\pi'$
 of $S\cup \neg C$
with $\delta(\pi')$ strictly smaller than $\delta(\pi)$
and that contradicts the minimality of $\delta(\pi)$.
We conclude that $\delta(\pi)=\emptyset$, and hence we have 
$S\vdash_{C\downarrow_\R} C$.
\end{proof}

\begin{proposition}{\label{prop:pro3}}
  Let $C$ be a ground clause. We have that $S\vdash_{C\downarrow_\R} C$ implies $S\models C$.
\end{proposition}
\begin{proof}
Let $C$ be a ground clause such that $S\vdash_{C\downarrow_\R} C$. 
This implies that there is a DAG proof of $S\cup \neg C$, and hence by correctness of the resolution,
$S\cup \neg C$ is unsatisfiable, and hence $S\models C$.
\end{proof}

\begin{proposition}{\label{prop:decidable:saturated}}
  Let $\R$ be a finite rewriting system on atoms, and $S$ be an 
 \R{}-saturated set of clauses. The ground entailment
  problem for $S$ is decidable.
\end{proposition}
\begin{proof}
Let $C$ be an arbitrary ground clause,  the \textit{ground  entailment problem} for $S$ is decidable
if and only if  $S\models C$ is decidable.
By the propositions \ref{prop:locality} and \ref{prop:pro3}, we have that 
$S\models C$ if and only if $S\vdash_{C\downarrow_\R} C$.
By Lemma \ref{lem:decidable:local:entailment}, $S\vdash_{C\downarrow_\R} C$ is decidable. We conclude that 
$S\models C$ is decidable, and hence the \textit{ground entailment problem}
is decidable.
\end{proof}

From the previous lemmas and propositions, we obviously deduce the following theorem which is 
 the main result of the  paper.

\begin{theorem}{\label{theo:main}}
  Let $\atomordering$ be a  well-founded,  monotone atom ordering such
  that  $A\prec_a B$
  implies $\Var{A}\subseteq\Var{B}$ for every atoms $A$ and $B$. Let $S$ be a set of clauses such that the 
  saturation on $S$ terminates using the atom ordering
  $\atomordering$. Then the ground entailment problem for $S$ is decidable.
\end{theorem}

\section{Comparison with existing works}\label{sec:comp}

This paper is meant to be an extension of~\cite{Basin-Ganzinger-01} to
more general orderings and it relies on \emph{a priori} instead of
\emph{a posteriori} ordered resolution used in \cite{Basin-Ganzinger-01}. 
Though various settings are
considered  in ~\cite{Basin-Ganzinger-01}, a common trait is that the
atom ordering $\prec_a$ and the term ordering $\prec_t$ satisfy the following:
\begin{itemize}
\item the term ordering $\prec_t$ is well-founded and  total on ground terms; 
\item the atom ordering $\prec_a$  is
  compatible with the term ordering  $\prec_t$,
  \emph{i.e.}
$A(s_1,\ldots,s_m) \prec_a
  B(t_1,\ldots,t_n)$
whenever for  any $1\le
  j\le m$
there exists $1\le i\le n$ such that 
 $s_j\prec_t t_i$;
\item the atom ordering  $\prec_a$  is monotone;
\item every term $t$ has only a finite number of smaller terms for $\prec_t$.
\end{itemize}
We prove below that such orderings also satisfy our criteria when the
underlying term ordering is subterm (\textit{i.e.} $u[t]\succ t$ for every terms $u$ and $t$), which is the case for term
orderings such as KBO, LPO, RPO, etc.

\begin{proposition}
  If there exists an infinite number of terms and if the term ordering
  $\prec_t$    is \emph{subterm} then under the above conditions $A\prec_a B$ implies
  $\Var{A}\subseteq\Var{B}$.
\end{proposition}

\begin{proof}
  Assume there exists a term $t$ such that there does not exist $t'$
  with $t\prec_t t'$. Since the ordering is total on ground terms for
  every term $t'\neq t$ we have $t'\prec_t t$. Since there exists an
  infinite number of ground terms this contradicts the assumption that
  every term has only a finite number of terms smaller than itself.
  Thus for every term $t$ there exists a term $t'$ with $t\prec_t t'$.

  Now let $A$ and $B$ be two atoms, and assume
  $\Var{A}\not\subseteq\Var{B}$. Let $\sigma$ be a substitution
  grounding $B$, \emph{i.e.}, $B\sigma=b(s_1,\ldots,s_m)$. Wlog
  assume that $s_1$ is maximal among the $s_1,\ldots,s_m$ for the term
  ordering $\prec_t$. Let $t$ be a term greater than $s_1$. Let us
  extend $\sigma$ on $\Var{A}\setminus \Var{B}$ by a substitution
  $\tau$ mapping every $x\in\Var{A}\setminus \Var{B}$ to $t$. Since
  there is at least one occurrence of one such $x$ in $A$, and since
  the ordering is subterm, there exists a term $t'$ in $A\sigma\tau$
  that contains $t$ as a subterm. Since the ordering is subterm this implies $t
  \prec_t t'$. Since the ordering on ground atoms is compatible with
  the ordering on ground terms this implies $B\sigma \prec_a
  A\sigma\tau$. Thus $\Var{A}\not\subseteq\Var{B}$ implies
  $A\not\prec_a B$.
\end{proof}

Finally the assumptions employed in~\cite{Basin-Ganzinger-01} to derive
complexity results imply that the number of atoms smaller than a given
ground atom of size $n$ is in $\mathcal{O}(f(n))$ and such atoms may be
enumerated in time $\mathcal{O}(g(n))$ for two computable functions $f$
and $g$. Since we do not assume the same finiteness property we cannot
directly state complexity results. However we note that there is a lot
of works on the complexity analysis of term rewriting systems. While
these works aim at bounding the maximal length of a derivation, we
believe that it could still be useful to provide theoretic upper
bounds on the number of atoms smaller than the atoms in a fixed set
$C$ for the constructed ordering $\prec_{\mathcal{R}}$.


\section{Conclusion}
\label{sec:conclusion}

We have presented in this paper an extension of a classical result by
Basin and Ganzinger~\cite{Basin-Ganzinger-01}. The relaxation of the
hypothesis on the ordering lead to a further extension for resolution
modulo an equational theory~\cite{Huet72,NR94,Vigneron94}. 
We note that the redundancy notion introduced in \cite{BachmairG90} is
based on an ordering of clauses as multisets of literals. A drawback of
the saturation procedure presented in this paper is that clauses are seen
as sets of literals; Thus we cannot apply as is their result of combination 
of saturation with subsumption.
We plan to prove in future works that it is possible to add to our saturation procedure a 
backward subsumption
rule while preserving the construction of the finite complexity atom ordering.

We believe the
technique employed can be extended to add a reflectivity or transitivity
axiom to an already saturated theory. Also, we thank Chris
Lynch~\cite{lynch-PC} for having pointed to us (by giving a
counter-example) that the method cannot be extended \textit{as is} to
superposition. Finally we believe that a consequence of our proof is
that saturated theories are complete for contextual
deduction~\cite{BR91,NO90}, which may help in the resolution of
\cite{rta-loop}, though further work is needed to confirm this conjecture.

\bibliography{bibliographie}
\end{document}